\documentclass[10pt]{article}
\usepackage[a4paper, total={6in, 8in}]{geometry}
\usepackage[usenames,dvipsnames,svgnames,table]{xcolor} % usenames : 16 basic colors,  dvipsnames : other 68 colors, svgnames : another 150 colors. It is used instead of \usepackage{color}. Can find descriptions for color in xcolor package. 

\usepackage{algorithm}
\usepackage{algorithmic}
\floatname{algorithm}{Algorithm}

\usepackage{longtable}
\usepackage{hhline}
\usepackage{float}
\newfloat{algorithm}{tb}{lop}

\usepackage{enumitem} % for \begin{description}
\usepackage{rotfloat} % For [H] option in sidewaystable
\usepackage{graphicx}
\usepackage{epsfig}
\usepackage{amssymb, amsmath, amsthm}
\usepackage{tikz}
\usetikzlibrary{3d, calc}
\usepackage{verbatim}
\usepackage{hyperref}
\usepackage{natbib}
\usepackage{authblk}
\usepackage{multirow}
\usepackage{subcaption} % for \begin{subtable}
\usepackage{array}
\newcolumntype{L}{>{\centering\arraybackslash}m{0.1\linewidth}}

\newtheorem{theorem}{Theorem}[section]
\newtheorem{lemma}[theorem]{Lemma}
\newtheorem{proposition}[theorem]{Proposition}

\DeclareMathOperator{\varop}{var}
\DeclareMathOperator{\covop}{cov}

\newcommand{\R}{\mathbb{R}}

\DeclareMathOperator{\tr}{tr}

\newcommand{\Ptwo}{\mathcal{P}_2(\R^d)}
\newcommand{\W}{W_2}
\newcommand{\iid}{\text{i.i.d.}}

\begin{document}
	
\title{On Heterogeneity in Wasserstein Space}
\author[1]{Kisung You\footnote{\url{kisung.you@baruch.cuny.edu}}}
\affil[1]{Department of Mathematics, Baruch College, City University of New York}
\date{}

\maketitle
\begin{abstract}
Data represented by probability measures arise as empirical distributions, posterior distributions, and feature-based representations of complex objects. We study heterogeneity in a population of probability measures through the expected value of a chosen transform of the pairwise Wasserstein distance. The resulting estimator is unbiased and, under simple moment conditions on the population law, is strongly consistent, asymptotically normal, and equipped with a consistent standard error. This also yields a simple comparison of two populations and remains stable under plug-in approximation when the measures are estimated. The associated empirical eccentricities identify the observations that contribute most strongly to heterogeneity within a sample.
\end{abstract}

%%%%%%%%%%%%%%%%%%%%%%%%%%%%%%%%%%%%%%%%%%%%%%%%%%%%%%%%%%%%%%%%%%%%%%%%%%%%%%%%%

\section{Introduction}

Statistical observations are increasingly represented by probability measures rather than vectors. Examples include empirical distributions built from repeated observations within a subject, posterior distributions attached to individual units, and measures extracted from images, texts, or point clouds through an embedding step. In such settings, heterogeneity is not a property of a single measure, but a property of the population law that governs the random measures.

Let $\Ptwo$ denote the set of Borel probability measures on $\R^d$ with finite second moment, equipped with the quadratic Wasserstein distance $\W$ \citep{villani_2003_TopicsOptimalTransportation,villani_2009_OptimalTransportOld}. If $\Pi$ is a Borel probability law on $\Ptwo$ and $\psi\colon [0,\infty)\to\R$ is measurable, define
\begin{equation}
D_{\psi}(\Pi)=E\bigl[\psi\{\W(\mu,\nu)\}\bigr],
\label{eq:Dpsi}
\end{equation}
where $\mu,\nu\stackrel{\iid}{\sim}\Pi$. We call $D_{\psi}(\Pi)$ the pairwise Wasserstein heterogeneity of $\Pi$. When $\psi(t)=t^2$, the functional measures mean squared separation under the Wasserstein geometry. When $\psi$ is bounded, it yields a more robust summary that downweights unusually distant pairs.

The appeal of \eqref{eq:Dpsi} is its simplicity. It depends only on pairwise Wasserstein distances and therefore does not require a Wasserstein barycenter or any other centring object that are expensive to attain \citep{agueh_2011_BarycentersWassersteinSpace, you_2025_WassersteinMedianProbability}. The same framework also accommodates different scientific goals through the choice of $\psi$. One may emphasize overall spread with a distance-like transform or suppress rare extreme pairs with a bounded one.

Given independent draws $\mu_1,\ldots,\mu_n\sim\Pi$, the natural estimator of \eqref{eq:Dpsi} is the order-two $U$--statistic
\begin{equation}
U_n=\binom{n}{2}^{-1}\sum_{1\leq i<j\leq n} h(\mu_i,\mu_j),\quad
h(\mu,\nu)=\psi\{\W(\mu,\nu)\}.
\label{eq:Un}
\end{equation}
If $E|h(\mu,\nu)|<\infty$, then $U_n$ is unbiased for $D_{\psi}(\Pi)$ by standard $U$--statistic theory \citep{hoeffding_1948_ClassStatisticsAsymptotically,serfling_1980_ApproximationTheoremsMathematical}. The practical questions are how to express this integrability in terms of transparent assumptions on $\Pi$, and how to translate the resulting large-sample theory into an inferential procedure that is easy to report and interpret.

This paper addresses both questions. Under mild Wasserstein moment conditions, the estimator in \eqref{eq:Un} is strongly consistent and asymptotically normal, and its asymptotic variance is consistently estimated by the sample variance of observation-level eccentricities. Those eccentricities are useful not only for standard errors. They also indicate which individual measures are unusually far from the sample, so they turn the heterogeneity summary into an interpretable diagnostic.

\section{Theory and inference}

We first relate integrability of the kernel $h$ to moment conditions on the population law.

\begin{lemma}[Moment domination]\label{lem:domination}
Suppose that, for some $p\geq 1$ and constants $a,b\geq 0$,
\[
|\psi(t)|\leq a+b t^p\quad (t\geq 0),
\]
and fix $\mu_0\in\Ptwo$. Then there exists $C\geq 0$ such that, for all $\mu,\nu\in\Ptwo$,
\[
|h(\mu,\nu)|\leq C\bigl[1+\W(\mu,\mu_0)^p+\W(\nu,\mu_0)^p\bigr].
\]
Consequently, if $E\{\W(\mu,\mu_0)^p\}<\infty$, then $E|h(\mu,\nu)|<\infty$; if $E\{\W(\mu,\mu_0)^{2p}\}<\infty$, then $E\{h(\mu,\nu)^2\}<\infty$. These moment conditions do not depend on the particular choice of $\mu_0$.
\end{lemma}

Lemma~\ref{lem:domination} separates the role of the transform from that of the population law. Once $\psi$ grows no faster than a polynomial, integrability of the kernel reduces to an ordinary Wasserstein moment condition on $\Pi$. Bounded transforms are especially forgiving, because then no moment assumption is needed at all.

Write
\begin{equation}
g(\mu)=E\{h(\mu,\mu_1)\}-D_{\psi}(\Pi),
\label{eq:g}
\end{equation}
where $\mu_1\sim\Pi$ is independent of $\mu$. The function $g$ is the first Hoeffding projection of the kernel \citep{hoeffding_1948_ClassStatisticsAsymptotically}. It compares the expected transformed distance from $\mu$ to a typical draw with the overall population average. Measures with $g(\mu)>0$ are more isolated than a typical observation on average, whereas measures with $g(\mu)<0$ are more central.

\begin{theorem}\label{thm:main}
Assume the growth condition of Lemma~\ref{lem:domination}.

\noindent (i) If $E\{\W(\mu,\mu_0)^p\}<\infty$ for some $\mu_0\in\Ptwo$, then $U_n\to D_{\psi}(\Pi)$ almost surely.

\noindent (ii) If $E\{\W(\mu,\mu_0)^{2p}\}<\infty$ for some $\mu_0\in\Ptwo$ and $\varop\{g(\mu_1)\}>0$, then
\[
n^{1/2}\{U_n-D_{\psi}(\Pi)\}\to \mathcal{N}\bigl[0,4\varop\{g(\mu_1)\}\bigr]
\]
in distribution.
\end{theorem}

Part (i) requires only integrability of the kernel, whereas part (ii) requires square integrability. This is the natural split since consistency depends on a first moment, while the Gaussian approximation depends on a second moment. If $\varop\{g(\mu_1)\}=0$, the kernel is first-order degenerate and the normal approximation does not capture the leading behaviour.

The projection in \eqref{eq:g} is not directly observable, but it has a natural sample analogue,
\begin{equation}
\widehat g_i=(n-1)^{-1}\sum_{j\neq i} h(\mu_i,\mu_j)-U_n\quad (i=1,\ldots,n).
\label{eq:ghat}
\end{equation}
We call $\widehat g_i$ the empirical eccentricity of $\mu_i$. It is the average transformed distance from $\mu_i$ to the rest of the sample, recentered by the sample-wide average. Large positive and negative values indicate unusually isolated measures and ones that are relatively representative of the sample cloud, respectively.

Because $\sum_{i=1}^n \widehat g_i=0$, the empirical eccentricities yield the following variance estimator
\begin{equation}
\widehat\sigma_n^2=4(n-1)^{-1}\sum_{i=1}^n \widehat g_i^2.
\label{eq:sighat}
\end{equation}

\begin{proposition}[Variance estimation]\label{prop:variance}
Under the assumptions of Theorem~\ref{thm:main}(ii), $\widehat\sigma_n^2\to 4\varop\{g(\mu_1)\}$ in probability. Hence
\[
\widehat\sigma_n^{-1} n^{1/2}\{U_n-D_{\psi}(\Pi)\}\to \mathcal{N}(0,1)
\]
in distribution.
\end{proposition}

Accordingly, a large-sample $(1-\alpha)$ interval for $D_{\psi}(\Pi)$ is
\[
U_n\pm z_{1-\alpha/2}\,\widehat\sigma_n n^{-1/2},
\]
and the null hypothesis $H_0\colon D_{\psi}(\Pi)=D_0$ may be tested with $\widehat\sigma_n^{-1}n^{1/2}(U_n-D_0)$. The same eccentricities also provide a compact diagnostic. A long positive tail in the sorted values of $\widehat g_i$ indicates that a few unusually isolated measures are driving the overall heterogeneity.

The same construction compares two populations of random measures. Let $\Pi_A$ and $\Pi_B$ be laws on $\Ptwo$, with independent samples of sizes $n_A$ and $n_B$, and let $U_A$ and $U_B$ be the corresponding within-group statistics.

\begin{proposition}[Two-sample comparison]\label{prop:twosample}
Assume that the conditions of Theorem~\ref{thm:main}(ii) hold for both groups, with non-zero projection variances. If $\Delta=D_{\psi}(\Pi_A)-D_{\psi}(\Pi_B)$, then $U_A-U_B\to \Delta$ almost surely and
\[
\frac{(U_A-U_B)-\Delta}{(\widehat\sigma_A^2/n_A+\widehat\sigma_B^2/n_B)^{1/2}}
\to \mathcal{N}(0,1)
\]
in distribution, where $\widehat\sigma_A^2$ and $\widehat\sigma_B^2$ are computed from \eqref{eq:sighat} within each group.
\end{proposition}

This yields an immediate test of equal heterogeneity across two populations. In particular, $H_0\colon D_{\psi}(\Pi_A)=D_{\psi}(\Pi_B)$ is rejected for large values of
\[
\frac{|U_A-U_B|}{(\widehat\sigma_A^2/n_A+\widehat\sigma_B^2/n_B)^{1/2}}.
\]
The test is useful when the scientific question is comparative rather than absolute, for example when one wishes to determine whether one class of random measures is more variable than another under the same transform.

In practice, the measures $\mu_i$ are often themselves estimated. The next result shows that preliminary estimation affects the heterogeneity statistic only through an average Wasserstein approximation error, provided that $\psi$ is Lipschitz.

\begin{proposition}[Plug-in stability]\label{prop:plugin}
Assume that $\psi$ is $L$-Lipschitz. Let $\widehat U_n$ be obtained from \eqref{eq:Un} by replacing $\mu_i$ with approximations $\widehat\mu_i$. Then
\[
|\widehat U_n-U_n|\leq 2Ln^{-1}\sum_{i=1}^n \W(\widehat\mu_i,\mu_i).
\]
\end{proposition}

Proposition~\ref{prop:plugin} shows that first-stage estimation enters only through an average Wasserstein error term. In particular, if $n^{-1}\sum_{i=1}^n \W(\widehat\mu_i,\mu_i)=o_p(n^{-1/2})$, then replacing $U_n$ by $\widehat U_n$ does not alter the first-order inference. Combined with available rates for empirical measures in Wasserstein distance \citep{fournier_2015_RateConvergenceWasserstein}, this justifies replacing idealized unit-level measures by empirical approximations.

\section{Interpretation in simple models}

The functional in \eqref{eq:Dpsi} becomes readily transparent in simple families. The next examples show how pairwise Wasserstein heterogeneity reduces to familiar quantities and how the empirical eccentricities should be interpreted.

\smallskip\noindent\textit{Random translations.}
Let $\mu_{\theta}=T_{\theta}\#\mu_0$, where $T_{\theta}(x)=x+\theta$ and $\mu_0\in\Ptwo$ is fixed. Every measure in the family is the same template shifted by a vector $\theta$, so translation equivariance of $\W$ gives
\[
\W(\mu_{\theta},\mu_{\theta'})=\|\theta-\theta'\|.
\]
If $\mu=\mu_{\Theta}$ for an $\R^d$-valued random shift $\Theta$, then heterogeneity of the random measures is exactly heterogeneity of the latent shifts
\[
D_{\psi}(\Pi)=E\bigl[\psi(\|\Theta-\Theta'\|)\bigr],
\]
where $\Theta$ and $\Theta'$ are independent copies. For $\psi(t)=t^2$, one can compute this quantity explicitly. Since
\[
\Theta-\Theta'=(\Theta-E\Theta)-(\Theta'-E\Theta'),
\]
expanding the squared norm yields
\[
E\|\Theta-\Theta'\|^2=E\|\Theta-E\Theta\|^2+E\|\Theta'-E\Theta'\|^2-2E\bigl[(\Theta-E\Theta)^\top (\Theta'-E\Theta')\bigr].
\]
The cross term vanishes by independence, and the first two terms are equal. Therefore
\[
E\|\Theta-\Theta'\|^2=2E\|\Theta-E\Theta\|^2=2\tr\{\covop(\Theta)\}.
\]
Thus $D_{t^2}(\Pi)$ is twice the total variance of the shift variable. In this model, a measure has large empirical eccentricity when its shift is unusually far from the centre of the latent location distribution.

\smallskip\noindent\textit{Posterior-induced Gaussian measures.}
Consider a conjugate model in which $Y_i\mid\theta\sim \mathcal{N}(\theta,1)$ independently and the prior is $\theta\sim \mathcal{N}(0,\tau_0^2)$. After observing $Y=(Y_1,\ldots,Y_m)$, the posterior is $\theta\mid Y\sim \mathcal{N}(\theta_m,\tau_m^2)$ for the usual posterior mean $\theta_m$ and posterior variance $\tau_m^2$,
\[
\tau_m^2=(\tau_0^{-2}+m)^{-1},\quad
\theta_m=\tau_m^2 m\bar Y .
\]
Each posterior draw $\theta$ induces the Gaussian measure $\mathcal{N}(\theta,1)$, so posterior uncertainty in $\theta$ becomes a population of probability measures.

For Gaussian measures with equal covariance matrices, the squared Wasserstein distance is the squared Euclidean distance between the means \citep{takatsu_2011_WassersteinGeometryGaussian}. Hence
\[
\W^2\bigl[\mathcal{N}(\theta,1),\mathcal{N}(\theta',1)\bigr]=(\theta-\theta')^2.
\]
If $\Pi(\cdot\mid Y)$ denotes the posterior law of the induced measures and $\psi(t)=t^2$, then
\[
D_{t^2}(\Pi\mid Y)=E\bigl[(\theta-\theta')^2\mid Y\bigr].
\]
Now $\theta-\theta'$ has conditional mean zero and conditional variance $2\tau_m^2$, because $\theta$ and $\theta'$ are independent posterior draws. Therefore
\[
D_{t^2}(\Pi\mid Y)=2\tau_m^2=2\varop(\theta\mid Y).
\]
So in this familiar Bayesian model the Wasserstein heterogeneity of the induced random measures is nothing more than twice the posterior variance. The corresponding sample estimator is available immediately. If $\theta^{(1)},\ldots,\theta^{(n)}$ are posterior draws, then
\[
U_n=\binom{n}{2}^{-1}\sum_{1\leq i<j\leq n}(\theta^{(i)}-\theta^{(j)})^2=2s_{\theta}^2,
\]
where $s_{\theta}^2$ is the sample variance of the posterior draws.

\smallskip\noindent\textit{A symmetric two-point population.}
Let $\Pi=\tfrac{1}{2}\delta_{\mu_1}+\tfrac{1}{2}\delta_{\mu_2}$, with $\mu_1\neq\mu_2$. When two independent draws are taken from this population, they coincide with probability $1/2$ and differ with probability $1/2$. It follows that
\[
D_{\psi}(\Pi)=\tfrac{1}{2}\psi(0)+\tfrac{1}{2}\psi\bigl[\W(\mu_1,\mu_2)\bigr].
\]
If additionally $\psi(0)=0$, then the heterogeneity is one half of the transformed separation between the two support points, so the population is genuinely heterogeneous whenever $\mu_1$ and $\mu_2$ are distinct.

The first-order projection tells a different story. For either support point $\mu_k$,
\[
E\{h(\mu_k,\mu')\}=\tfrac{1}{2}\psi(0)+\tfrac{1}{2}\psi\bigl[\W(\mu_1,\mu_2)\bigr]=D_{\psi}(\Pi),
\]
where $\mu'\sim\Pi$. Thus $g(\mu_k)=0$ for both $k=1,2$. In words, neither support point is more eccentric than the other as each sees exactly the same expected distance profile. Nevertheless, the population has positive heterogeneity, implying that the heterogeneity is generated entirely by a second-order effect rather than by first-order variation in the projection. This is the simplest example showing why degeneracy can coexist with a non-zero target and why the Gaussian approximation in Theorem~\ref{thm:main} is not always appropriate.

These examples also clarify how \eqref{eq:ghat} should be read in practice. The empirical eccentricities are observation-level summaries of how strongly individual measures contribute to the overall heterogeneity. They therefore play an interpretive role as well as an inferential one.

\section{Numerical illustrations}

\subsection{Synthetic Gaussian measures}
\label{sec:synthetic}

We begin with a setting in which the quadratic Wasserstein distance is available in closed form. Each observation is an axis-aligned Gaussian measure
\[
\mu_i=\mathcal{N}\bigl\{m_i,\mathrm{diag}(s_{i1}^2,s_{i2}^2)\bigr\},
\]
where $m_i\in\R^2$ is a mean vector and $(s_{i1},s_{i2})$ is the vector of marginal standard deviations. For this family,
\[
W_2^2(\mu_i,\mu_j)=\|m_i-m_j\|^2+\|s_i-s_j\|^2,
\]
so all pairwise Wasserstein distances are exact. This example allows to check  both the inferential part of the method and the observation-wise interpretation of the empirical eccentricities.

We generated three groups, each of size 60. Group A is tightly concentrated and has little variation in either location or scale. Group B is more diffuse. Group C is a mixture, with a compact main component and a smaller separated component. Thus the design contains both broad within-group variation and a localized source of atypicality.

\begin{figure}[ht]
\centering
\includegraphics[width=\textwidth]{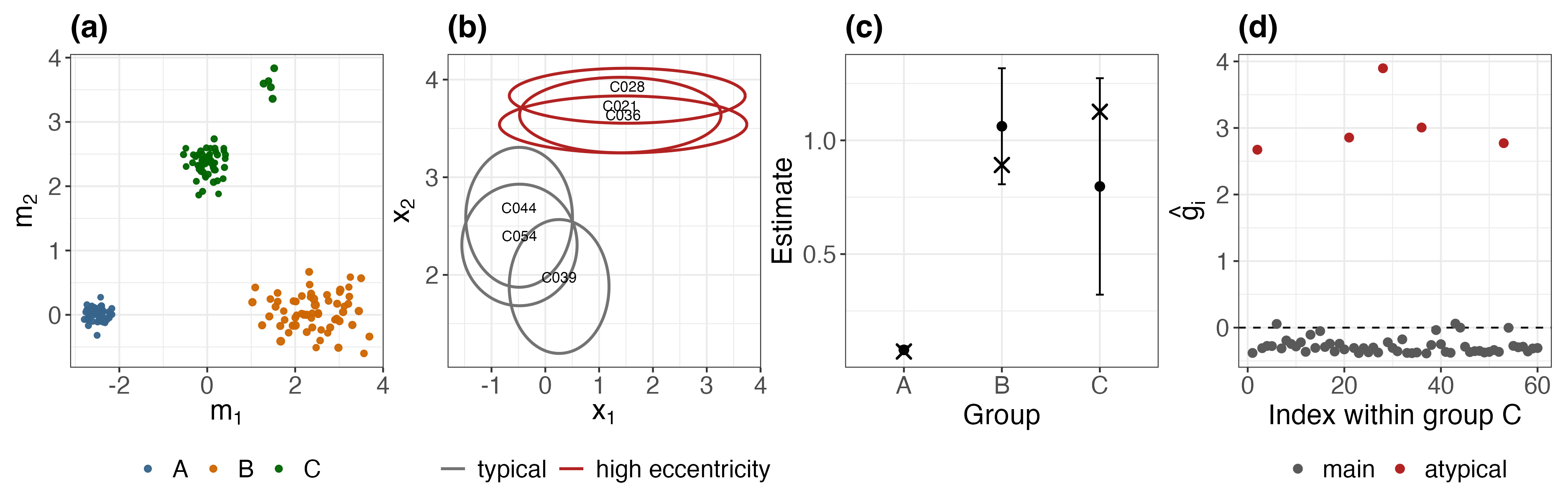}
\caption{Synthetic Gaussian experiment. Panel (a) shows the locations of the generated measures; colours indicate the three groups A, B, and C. Panel (b) shows selected Gaussian measures from group C; grey ellipses denote typical observations and red ellipses denote observations with large empirical eccentricity. Panel (c) shows within-group heterogeneity estimates (points) with 95\% Wald intervals (vertical bars); crosses indicate analytic reference values. Panel (d) shows empirical eccentricities $\widehat g_i$ for observations in group C; red points correspond to the atypical component and dark grey points to the main component.}
\label{fig:synthetic-overview}
\end{figure}

For the squared transform $\psi(t)=t^2$, the corresponding heterogeneity targets are available analytically. Figure~\ref{fig:synthetic-overview} summarizes one simulated data set. Panel (a) shows the intended geometric structure that group A is concentrated, group B is diffuse, and group C contains a separated minor component. Panel (b) shows that the observations with the largest empirical eccentricities lie in that separated component, which is precisely the behaviour expected from \eqref{eq:ghat}. Panel (c) shows that the within-group estimates are close to the analytic targets, with noticeably wider uncertainty for group C, reflecting the extra variability induced by the mixture structure. Panel (d) gives the same message at the observation level: large positive eccentricities are concentrated in the atypical component, whereas the main component lies near zero.

%A Monte Carlo study, reported in the Supplementary Material, confirms the same pattern over repeated data sets: the estimator is nearly unbiased, the Wald intervals are close to nominal, and the two-sample comparison of Proposition~\ref{prop:twosample} behaves as predicted by the theory.

\begin{figure}[ht]
\centering
\includegraphics[width=.7\textwidth]{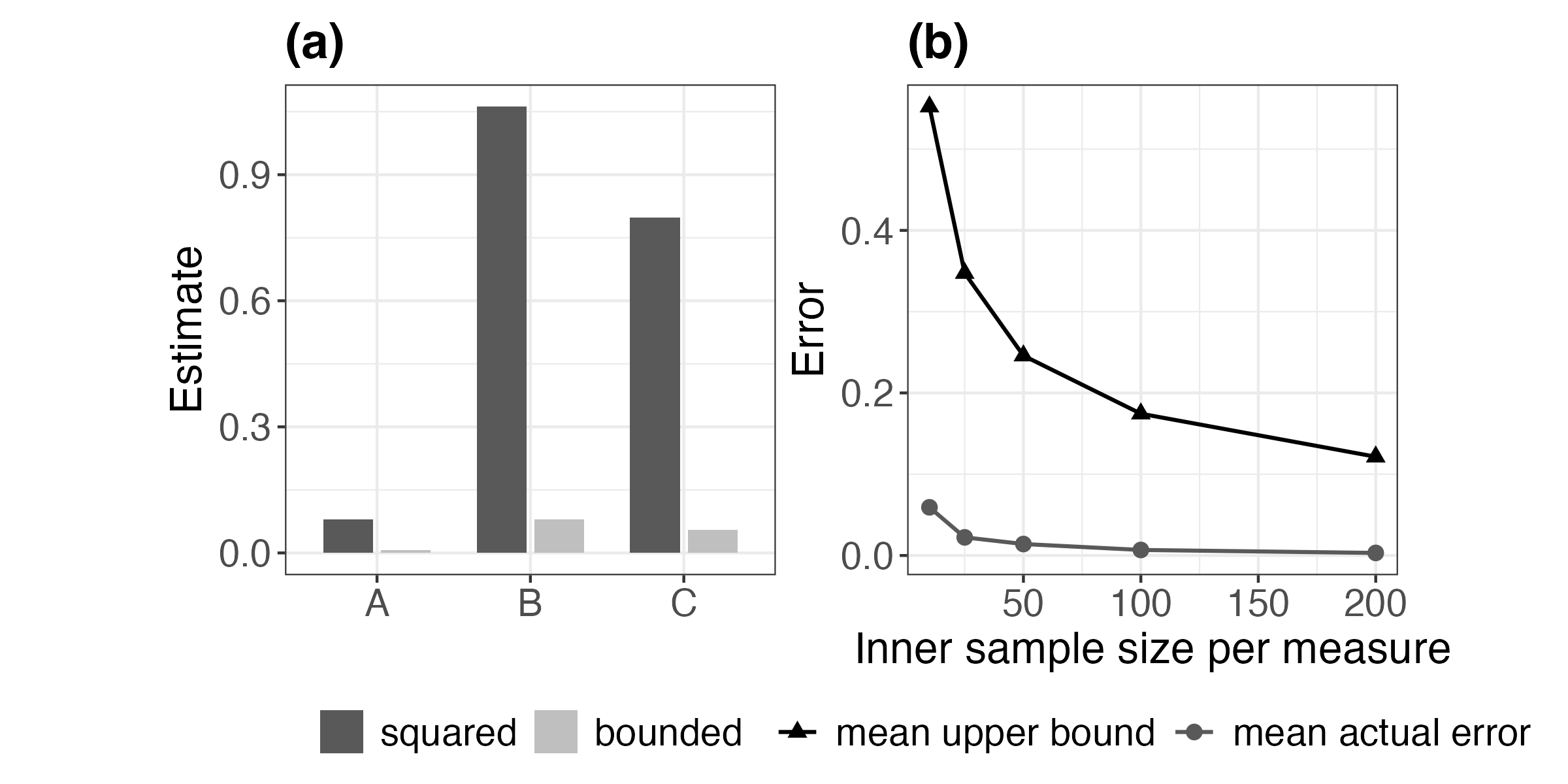}
\caption{Transform choice and plug-in stability in the synthetic experiment.}
\label{fig:synthetic-secondary}
\end{figure}

Figure~\ref{fig:synthetic-secondary} highlights two further features.  Panel (a) compares the squared transform $\psi(t)=t^2$ with the bounded transform $\psi(t)=t^2/(t^2+c_0^2)$, where $c_0$ is the sample median of the pairwise Wasserstein distances $\{W_2(\mu_i,\mu_j):i<j\}$. As expected, the bounded transform compresses the heterogeneity estimates, mostly for the groups with larger pairwise distances. Panel (b) examines Proposition~\ref{prop:plugin} for the Lipschitz transform $\psi(t)=t$. The observed plug-in error decreases as the inner sample size increases and remains below the theoretical upper bound throughout, which is consistent with the stated stability result.

Taken together, Figs.~\ref{fig:synthetic-overview} and \ref{fig:synthetic-secondary} illustrate three  points. The heterogeneity estimator and its standard error behave as expected in a model with known Wasserstein structure. The empirical eccentricities identify the observations that drive the heterogeneity signal. The choice of transform and the use of plug-in approximations have the anticipated robustness and stability effects.

\subsection{Handwritten digits}

We next illustrate the method with the MNIST handwritten digits data \citep{lecun_1998_MNISTDatabaseHandwritten}. For each digit $k\in\{0,\ldots,9\}$, we randomly selected 500 images, where each image is normalized as a probability measure on the pixel grid. Using the precomputed within-digit pairwise distances, we estimated the heterogeneity functional with the squared transform $\psi(t)=t^2$. Standard errors and 95\% Wald intervals were computed from the empirical eccentricities in Section~2.

\begin{table}[t]
\centering
\begin{tabular}{lcccccccccc}
\hline
Digit & 0 & 1 & 2 & 3 & 4 & 5 & 6 & 7 & 8 & 9 \\
\hline
Estimate    & 4.939 & 5.916 & 9.241 & 8.028 & 7.828 & 9.601 & 6.743 & 8.788 & 6.493 & 6.526 \\
SE          & 0.125 & 0.272 & 0.203 & 0.216 & 0.194 & 0.223 & 0.192 & 0.270 & 0.190 & 0.204 \\
95\% lower  & 4.694 & 5.383 & 8.844 & 7.604 & 7.448 & 9.164 & 6.366 & 8.258 & 6.120 & 6.127 \\
95\% upper  & 5.184 & 6.449 & 9.638 & 8.451 & 8.208 & 10.038 & 7.120 & 9.317 & 6.866 & 6.925 \\
Rank        & 10 & 9 & 2 & 4 & 5 & 1 & 6 & 3 & 8 & 7 \\
\hline
\end{tabular}
\caption{Within-digit Wasserstein heterogeneity in MNIST. SE denotes the standard error. Rank 1 corresponds to the largest estimated within-digit heterogeneity.}
\label{tab:mnist-summary}
\end{table}

Table~\ref{tab:mnist-summary} summarizes the results. All ten digits exhibit substantial within-digit heterogeneity, and even the smallest interval lies well above zero. The main question is therefore comparative rather than absolute, namely which digit classes are written in relatively stable ways, which admit several distinct forms, and which individual images deviate most from the rest within each digit.

The largest point estimates are attained by digits 5, 2, and 7, whereas digits 0 and 1 are markedly less heterogeneous. This ordering accords with visual intuition that some digits are constrained by a relatively stable overall shape, while others admit several common constructions that differ in slant, curvature, or the balance between upper and lower strokes. At the same time, the confidence intervals for digits 5, 2, and 7 overlap, so it is more sensible to regard them as a high-heterogeneity group than to insist on a strict ranking among them.

\begin{figure}[ht]
\centering
\includegraphics[width=.75\textwidth]{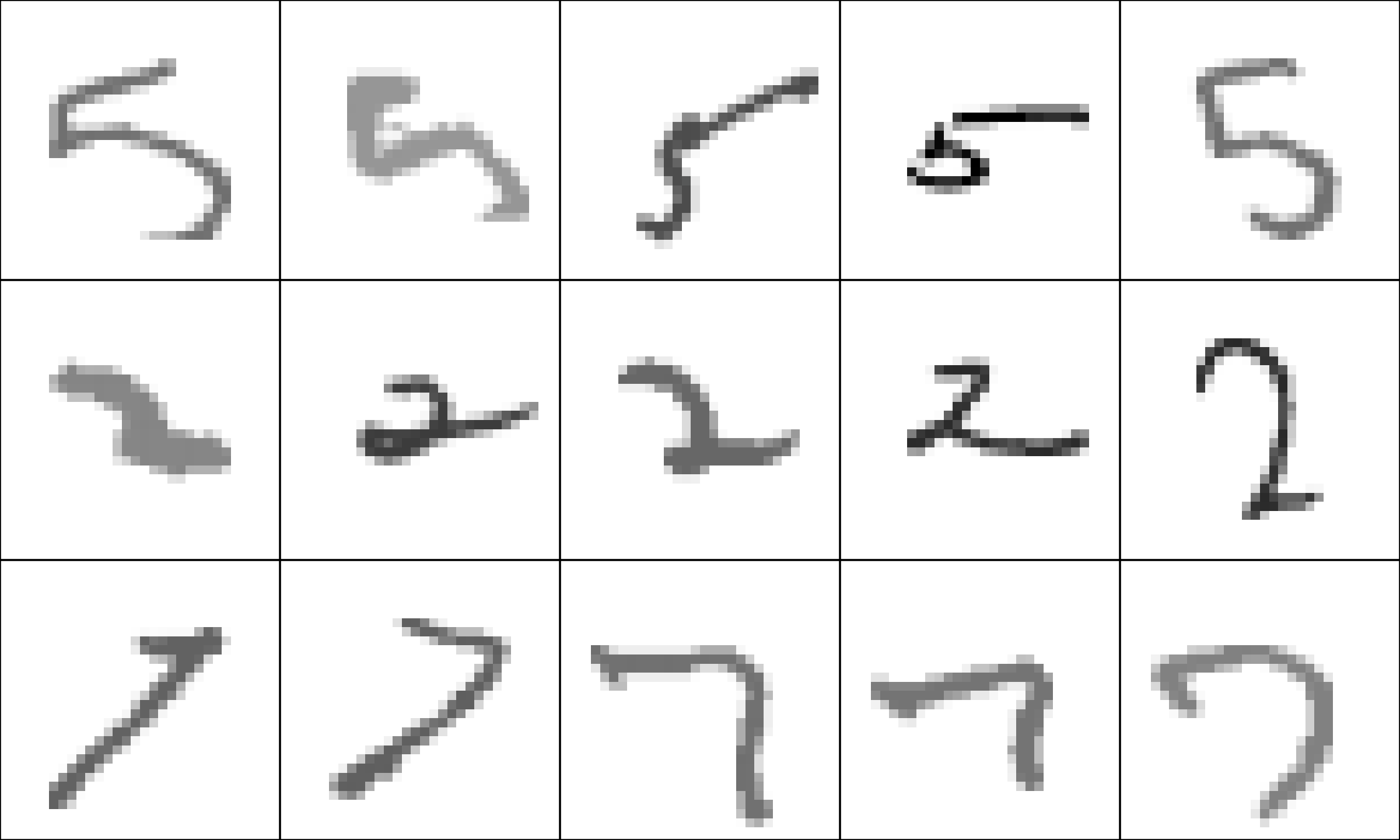}
\caption{Most eccentric images for digits 5 (top row), 2 (middle row), and 7 (bottom row).}
\label{fig:mnist-eccentric}
\end{figure} 

Figure~\ref{fig:mnist-eccentric} displays the most eccentric images from digits 5, 2, and 7, corresponding to the largest values of $\widehat g_i$, which contribute heavily to the corresponding $U$–statistics and illustrate how the empirical eccentricities identify the observations responsible for within-class heterogeneity.

\section{Discussion}

Pairwise transformed Wasserstein distance provides a simple population-level summary of spread for random probability measures. Once the kernel is controlled by Wasserstein moments, standard order-two $U$--statistic theory yields consistency, asymptotic normality, and a consistent standard error. The empirical eccentricities play a dual role to estimate the asymptotic variance and to identify the observations that contribute most strongly to heterogeneity.

The numerical illustrations highlight both aspects of the framework. In the synthetic Gaussian experiment, the estimator and its Wald intervals behave in line with the theory, while the empirical eccentricities recover the separated component responsible for the additional heterogeneity. In the handwritten-digit example, the same quantities rank digit classes by within-class variability and identify the images that make digits such as 5, 2, and 7 especially heterogeneous.

Two issues merit further study. The first is computation, since a full analysis requires all pairwise Wasserstein distances within a sample. The second is degeneracy: highly symmetric populations can have positive heterogeneity but vanishing first-order projection variance, so that the normal approximation in Theorem~\ref{thm:main} no longer applies. Natural next steps are scalable approximations based on incomplete pairwise sampling or faster Wasserstein computation, together with a more systematic study of how the choice of transform $\psi$ trades sensitivity against robustness.

\bibliographystyle{dcu}
\bibliography{references}

%%%%%%%%%%%%%%%%%%%%%%%%%%%%%%%%%%%%%%%%%%%%%%%%%%%%%%%%%%%%%%%%%%%%%%%%%%%%%%%%%
\newpage 
\appendix
%this is appendix
\section*{Proofs of Theoretical Results}

Throughout, write
\[
\theta = D_{\psi}(\Pi),\quad h_{ij}=h(\mu_i,\mu_j)=\psi\{\W(\mu_i,\mu_j)\}.
\]
Under the conditions of Theorem~\ref{thm:main}(ii), define
\[
g_i=g(\mu_i),\quad
\widetilde h_{ij}=h_{ij}-\theta-g_i-g_j.
\]
Then
\[
E(g_1)=0,
\]
and, for $k=1,2$,
\[
E\{\widetilde h_{12}\mid \mu_k\}=0
\]
almost surely.

\vspace{1cm}
\noindent \textbf{Proof of Lemma~\ref{lem:domination}}\\
By the triangle inequality for $\W$,
\[
\W(\mu,\nu)\leq \W(\mu,\mu_0)+\W(\nu,\mu_0).
\]
Since $p\geq 1$, the elementary inequality $(x+y)^p\leq 2^{p-1}(x^p+y^p)$ holds for all $x,y\geq 0$. Therefore
\[
\W(\mu,\nu)^p
\leq 2^{p-1}\bigl\{\W(\mu,\mu_0)^p+\W(\nu,\mu_0)^p\bigr\}.
\]
Using the growth condition on $\psi$,
\[
|h(\mu,\nu)|
=
\bigl|\psi\{\W(\mu,\nu)\}\bigr|
\leq
a+b\,\W(\mu,\nu)^p
\leq
a+b\,2^{p-1}\bigl\{\W(\mu,\mu_0)^p+\W(\nu,\mu_0)^p\bigr\}.
\]
Hence, with $C=a+b\,2^{p-1}$,
\[
|h(\mu,\nu)|
\leq
C\bigl\{1+\W(\mu,\mu_0)^p+\W(\nu,\mu_0)^p\bigr\},
\]
which proves the stated domination bound.

Now suppose that $E\{\W(\mu,\mu_0)^p\}<\infty$. Since $\mu$ and $\nu$ are independent copies,
\[
E|h(\mu,\nu)|
\leq
C\Bigl[1+E\{\W(\mu,\mu_0)^p\}+E\{\W(\nu,\mu_0)^p\}\Bigr]
=
C\Bigl[1+2E\{\W(\mu,\mu_0)^p\}\Bigr]
<\infty.
\]

Next suppose that $E\{\W(\mu,\mu_0)^{2p}\}<\infty$. Squaring the domination bound and using
\[
(x+y+z)^2\leq 3(x^2+y^2+z^2),
\]
we obtain
\[
h(\mu,\nu)^2
\leq
3C^2\bigl\{1+\W(\mu,\mu_0)^{2p}+\W(\nu,\mu_0)^{2p}\bigr\}.
\]
Taking expectations gives
\[
E\{h(\mu,\nu)^2\}
\leq
3C^2\Bigl[1+2E\{\W(\mu,\mu_0)^{2p}\}\Bigr]
<\infty.
\]

It remains to show that the finiteness of the Wasserstein moments does not depend on the choice of $\mu_0$. Let $\mu_1\in\Ptwo$ be fixed, and let $q>0$. Again by the triangle inequality,
\[
\W(\mu,\mu_1)\leq \W(\mu,\mu_0)+\W(\mu_0,\mu_1).
\]
Hence
\[
\W(\mu,\mu_1)^q
\leq
2^{q-1}\bigl\{\W(\mu,\mu_0)^q+\W(\mu_0,\mu_1)^q\bigr\}.
\]
Since $\mu_0$ and $\mu_1$ are fixed elements of $\Ptwo$, the quantity $\W(\mu_0,\mu_1)$ is finite. Therefore
\[
E\{\W(\mu,\mu_1)^q\}
\leq
2^{q-1}E\{\W(\mu,\mu_0)^q\}
+
2^{q-1}\W(\mu_0,\mu_1)^q.
\]
Thus finiteness of $E\{\W(\mu,\mu_0)^q\}$ implies finiteness of $E\{\W(\mu,\mu_1)^q\}$. Interchanging the roles of $\mu_0$ and $\mu_1$ gives the converse implication. Applying this with $q=p$ and $q=2p$ proves the claim.\qed

\vspace{1cm}
\noindent \textbf{Proof of Theorem~\ref{thm:main}}\\
For part (i), Lemma~\ref{lem:domination} implies that
\[
E|h(\mu_1,\mu_2)|<\infty.
\]
The kernel $h$ is measurable and symmetric, because $\W(\mu,\nu)=\W(\nu,\mu)$. Therefore the classical strong law for $U$--statistics of order two applies, and yields
\[
U_n\to \theta=D_{\psi}(\Pi)
\]
almost surely. For more details, see,  for example, \citet{hoeffding_1948_ClassStatisticsAsymptotically} or \citet{serfling_1980_ApproximationTheoremsMathematical}.

We now prove part (ii). By Lemma~\ref{lem:domination},
\[
E(h_{12}^2)<\infty.
\]
Since
\[
\theta+g(\mu_1)=E(h_{12}\mid \mu_1),
\]
conditional Jensen gives
\[
E\bigl[\{\theta+g(\mu_1)\}^2\bigr]
=
E\Bigl[\bigl\{E(h_{12}\mid \mu_1)\bigr\}^2\Bigr]
\leq
E(h_{12}^2)
<\infty.
\]
Hence $E\{g(\mu_1)^2\}<\infty$. It follows that $E(\widetilde h_{12}^2)<\infty$ as well.

By construction,
\[
h_{ij}=\theta+g_i+g_j+\widetilde h_{ij},
\]
so the Hoeffding decomposition is
\[
U_n-\theta
=
\frac{2}{n}\sum_{i=1}^n g_i
+
U_n^{\circ},
\qquad
U_n^{\circ}
=
\binom{n}{2}^{-1}\sum_{1\leq i<j\leq n}\widetilde h_{ij}.
\]
We first show that the degenerate term $U_n^{\circ}$ is negligible on the $n^{1/2}$ scale. Since $E(U_n^{\circ})=0$, it is enough to compute its second moment. Consider two pairs $(i,j)$ and $(k,\ell)$ with $i<j$ and $k<\ell$.

If $\{i,j\}\cap\{k,\ell\}=\varnothing$, then $\widetilde h_{ij}$ and $\widetilde h_{k\ell}$ are independent and centred, so
\[
E(\widetilde h_{ij}\widetilde h_{k\ell})=0.
\]
If the two pairs overlap in exactly one index, say $(i,j)=(1,2)$ and $(k,\ell)=(1,3)$, then
\[
E(\widetilde h_{12}\widetilde h_{13})
=
E\Bigl[E(\widetilde h_{12}\widetilde h_{13}\mid \mu_1)\Bigr]
=
E\Bigl[E(\widetilde h_{12}\mid \mu_1)\,E(\widetilde h_{13}\mid \mu_1)\Bigr]
=
0,
\]
because $E(\widetilde h_{12}\mid \mu_1)=0$ almost surely. The same argument applies to any overlapping pair. Therefore all cross terms vanish, and only the diagonal terms remain:
\[
E\bigl\{(U_n^{\circ})^2\bigr\}
=
\binom{n}{2}^{-2}
\binom{n}{2}
E(\widetilde h_{12}^2)
=
\binom{n}{2}^{-1}E(\widetilde h_{12}^2)
=
\frac{2E(\widetilde h_{12}^2)}{n(n-1)}.
\]
Hence
\[
E\bigl\{n(U_n^{\circ})^2\bigr\}
=
\frac{2E(\widetilde h_{12}^2)}{n-1}\to 0,
\]
so
\[
n^{1/2}U_n^{\circ}\to 0
\]
in $L^2$, and therefore in probability.

Now $(g_i)_{i\geq 1}$ are independent and identically distributed, with mean zero and variance $\operatorname{var}\{g(\mu_1)\}\in(0,\infty)$. The classical central limit theorem gives
\[
\frac{2}{n^{1/2}}\sum_{i=1}^n g_i
\to
\mathcal{N}\bigl[0,4\operatorname{var}\{g(\mu_1)\}\bigr]
\]
in distribution. Combining this with the Hoeffding decomposition and Slutsky's theorem yields
\[
n^{1/2}\{U_n-\theta\}
=
\frac{2}{n^{1/2}}\sum_{i=1}^n g_i
+
n^{1/2}U_n^{\circ}
\to
\mathcal{N}\bigl[0,4\operatorname{var}\{g(\mu_1)\}\bigr]
\]
in distribution, which proves part (ii).\qed

\vspace{1cm}
\noindent \textbf{Proof of Proposition~\ref{prop:variance}}\\
Let
\[
\bar g_n=n^{-1}\sum_{i=1}^n g_i,
\qquad
U_n^{\circ}
=
\binom{n}{2}^{-1}\sum_{1\leq i<j\leq n}\widetilde h_{ij}.
\]
Using the decomposition $h_{ij}=\theta+g_i+g_j+\widetilde h_{ij}$ and the identity
\[
U_n=\theta+\frac{2}{n}\sum_{i=1}^n g_i+U_n^{\circ},
\]
we obtain
\begin{align*}
\widehat g_i
&=
\frac{1}{n-1}\sum_{j\neq i}h_{ij}-U_n\\
&=
\frac{1}{n-1}\sum_{j\neq i}(\theta+g_i+g_j+\widetilde h_{ij})
-
\theta-\frac{2}{n}\sum_{k=1}^n g_k-U_n^{\circ}\\
&=
\frac{n-2}{n-1}(g_i-\bar g_n)
+
r_{i,n},
\end{align*}
where
\[
r_{i,n}
=
a_{i,n}-U_n^{\circ},
\qquad
a_{i,n}
=
\frac{1}{n-1}\sum_{j\neq i}\widetilde h_{ij}.
\]

We first show that the average squared remainder is negligible. Define
\[
A_n=\frac{1}{n-1}\sum_{i=1}^n r_{i,n}^2.
\]
Using $(x-y)^2\leq 2x^2+2y^2$,
\[
A_n
\leq
\frac{2}{n-1}\sum_{i=1}^n a_{i,n}^2
+
\frac{2n}{n-1}(U_n^{\circ})^2.
\]
Fix $i$. Then
\[
E(a_{i,n}^2)
=
\frac{1}{(n-1)^2}
\sum_{j\neq i}\sum_{k\neq i}E(\widetilde h_{ij}\widetilde h_{ik}).
\]
If $j\neq k$, then
\[
E(\widetilde h_{ij}\widetilde h_{ik})
=
E\Bigl[E(\widetilde h_{ij}\widetilde h_{ik}\mid \mu_i)\Bigr]
=
E\Bigl[E(\widetilde h_{ij}\mid \mu_i)\,E(\widetilde h_{ik}\mid \mu_i)\Bigr]
=
0.
\]
Therefore only the diagonal terms remain, and
\[
E(a_{i,n}^2)
=
\frac{1}{(n-1)^2}
\sum_{j\neq i}E(\widetilde h_{ij}^2)
=
\frac{E(\widetilde h_{12}^2)}{n-1}.
\]
From the proof of Theorem~\ref{thm:main},
\[
E\bigl\{(U_n^{\circ})^2\bigr\}
=
\frac{2E(\widetilde h_{12}^2)}{n(n-1)}.
\]
Hence
\[
E(A_n)
\leq
\frac{2n}{(n-1)^2}E(\widetilde h_{12}^2)
+
\frac{4}{(n-1)^2}E(\widetilde h_{12}^2)
\to 0.
\]
Thus
\[
A_n\to 0
\]
in $L^1$, and therefore in probability.

Now let
\[
\alpha_n=\frac{n-2}{n-1},
\qquad
b_{i,n}=\alpha_n(g_i-\bar g_n),
\qquad
B_n=\frac{1}{n-1}\sum_{i=1}^n b_{i,n}^2.
\]
Then
\[
B_n
=
\alpha_n^2\left\{
\frac{1}{n-1}\sum_{i=1}^n g_i^2
-
\frac{n}{n-1}\bar g_n^2
\right\}.
\]
Because $E(g_1)=0$ and $E(g_1^2)<\infty$, the strong law of large numbers gives
\[
\frac{1}{n}\sum_{i=1}^n g_i^2\to E(g_1^2),
\qquad
\bar g_n\to 0
\]
almost surely. Since $\alpha_n\to 1$, it follows that
\[
B_n\to E(g_1^2)=\operatorname{var}(g_1)
\]
almost surely, and hence in probability.

Finally,
\[
\widehat g_i=b_{i,n}+r_{i,n},
\]
so
\begin{align*}
\left|
\frac{1}{n-1}\sum_{i=1}^n \widehat g_i^2 - B_n
\right|
&=
\left|
\frac{1}{n-1}\sum_{i=1}^n (2b_{i,n}r_{i,n}+r_{i,n}^2)
\right|\\
&\leq
\frac{2}{n-1}\sum_{i=1}^n |b_{i,n}r_{i,n}|
+
\frac{1}{n-1}\sum_{i=1}^n r_{i,n}^2\\
&\leq
2B_n^{1/2}A_n^{1/2}
+
A_n,
\end{align*}
where the last step uses Cauchy--Schwarz. Since $B_n=O_p(1)$ and $A_n\to 0$ in probability,
\[
\frac{1}{n-1}\sum_{i=1}^n \widehat g_i^2
\to
\operatorname{var}(g_1)
\]
in probability. Multiplying by 4 gives
\[
\widehat\sigma_n^2
=
\frac{4}{n-1}\sum_{i=1}^n \widehat g_i^2
\to
4\operatorname{var}\{g(\mu_1)\}
\]
in probability.

The studentized central limit theorem now follows from Theorem~\ref{thm:main}(ii), the positivity assumption $\operatorname{var}\{g(\mu_1)\}>0$, and Slutsky's theorem:
\[
\frac{n^{1/2}\{U_n-\theta\}}{\widehat\sigma_n}
\to \mathcal{N}(0,1)
\]
in distribution.\qed

\vspace{1cm}
\noindent \textbf{Proof of Proposition~\ref{prop:twosample}}\\
Write
\[
\theta_A=D_{\psi}(\Pi_A),\quad
\theta_B=D_{\psi}(\Pi_B),\quad
\Delta=\theta_A-\theta_B,
\]
and let
\[
\sigma_A^2=4\operatorname{var}\{g_A(\mu_{A1})\},
\qquad
\sigma_B^2=4\operatorname{var}\{g_B(\mu_{B1})\},
\]
where $g_A$ and $g_B$ are the Hoeffding projections in the two groups.

By Theorem~\ref{thm:main}(i), applied separately to the two independent samples,
\[
U_A\to \theta_A,
\qquad
U_B\to \theta_B
\]
almost surely. Therefore
\[
U_A-U_B\to \Delta
\]
almost surely.

For the asymptotic normality, define
\[
Z_{A,n}=\frac{U_A-\theta_A}{\sigma_A/\sqrt{n_A}},
\qquad
Z_{B,n}=\frac{U_B-\theta_B}{\sigma_B/\sqrt{n_B}}.
\]
By Theorem~\ref{thm:main}(ii),
\[
Z_{A,n}\to \mathcal{N}(0,1),
\qquad
Z_{B,n}\to \mathcal{N}(0,1)
\]
in distribution. Because the two samples are independent, $Z_{A,n}$ and $Z_{B,n}$ are independent for each $n$. Hence, for any $s,t\in\R$,
\[
E\exp\{i(sZ_{A,n}+tZ_{B,n})\}
=
E(e^{isZ_{A,n}})\,E(e^{itZ_{B,n}})
\to
\exp\{-(s^2+t^2)/2\}.
\]
Thus
\[
(Z_{A,n},Z_{B,n})\to \mathcal{N}_2(0,I_2)
\]
in distribution.

Now let
\[
d_n^2=\frac{\sigma_A^2}{n_A}+\frac{\sigma_B^2}{n_B},
\qquad
c_{A,n}=\frac{\sigma_A/\sqrt{n_A}}{d_n},
\qquad
c_{B,n}=\frac{\sigma_B/\sqrt{n_B}}{d_n}.
\]
Then $c_{A,n}^2+c_{B,n}^2=1$, and
\[
T_n^{\circ}
:=
\frac{(U_A-U_B)-\Delta}{d_n}
=
c_{A,n}Z_{A,n}-c_{B,n}Z_{B,n}.
\]
The sequence $(c_{A,n},c_{B,n})$ lies in the compact unit circle, so every subsequence has a further subsequence along which
\[
(c_{A,n},c_{B,n})\to (c_A,c_B)
\]
with $c_A^2+c_B^2=1$. Along that further subsequence, the continuous mapping theorem gives
\[
T_n^{\circ}\to c_A Z_A-c_B Z_B,
\]
where $(Z_A,Z_B)\sim \mathcal{N}_2(0,I_2)$. Since $c_A^2+c_B^2=1$, the limiting distribution is $\mathcal{N}(0,1)$. Every subsequence therefore has a further subsequence converging to $\mathcal{N}(0,1)$, and it follows that
\[
T_n^{\circ}\to \mathcal{N}(0,1)
\]
in distribution.

It remains to replace $d_n$ by the estimated denominator
\[
\widehat d_n^2=\frac{\widehat\sigma_A^2}{n_A}+\frac{\widehat\sigma_B^2}{n_B}.
\]
By Proposition~\ref{prop:variance},
\[
\frac{\widehat\sigma_A^2}{\sigma_A^2}\to 1,
\qquad
\frac{\widehat\sigma_B^2}{\sigma_B^2}\to 1
\]
in probability. Also,
\[
\frac{\widehat d_n^2}{d_n^2}
=
w_{A,n}\frac{\widehat\sigma_A^2}{\sigma_A^2}
+
w_{B,n}\frac{\widehat\sigma_B^2}{\sigma_B^2},
\]
where
\[
w_{A,n}
=
\frac{\sigma_A^2/n_A}{\sigma_A^2/n_A+\sigma_B^2/n_B},
\qquad
w_{B,n}=1-w_{A,n}.
\]
Since $0\leq w_{A,n},w_{B,n}\leq 1$, the right-hand side is a convex combination of two quantities converging in probability to 1, so
\[
\frac{\widehat d_n^2}{d_n^2}\to 1
\]
in probability. By the continuous mapping theorem,
\[
\frac{\widehat d_n}{d_n}\to 1
\]
in probability. Slutsky's theorem now yields
\[
\frac{(U_A-U_B)-\Delta}
{(\widehat\sigma_A^2/n_A+\widehat\sigma_B^2/n_B)^{1/2}}
\to \mathcal{N}(0,1)
\]
in distribution.\qed

\vspace{1cm}
\noindent \textbf{Proof of Proposition~\ref{prop:plugin}}\\
Write
\[
\widehat h_{ij}
=
\psi\bigl\{\W(\widehat\mu_i,\widehat\mu_j)\bigr\},
\qquad
h_{ij}
=
\psi\bigl\{\W(\mu_i,\mu_j)\bigr\}.
\]
Then
\[
\widehat U_n-U_n
=
\binom{n}{2}^{-1}\sum_{1\leq i<j\leq n}(\widehat h_{ij}-h_{ij}).
\]
Hence
\[
|\widehat U_n-U_n|
\leq
\binom{n}{2}^{-1}\sum_{1\leq i<j\leq n}|\widehat h_{ij}-h_{ij}|.
\]
Since $\psi$ is $L$-Lipschitz,
\[
|\widehat h_{ij}-h_{ij}|
\leq
L\left|
\W(\widehat\mu_i,\widehat\mu_j)-\W(\mu_i,\mu_j)
\right|.
\]
Now use the reverse triangle inequality for the metric $\W$:
\[
\left|
\W(\widehat\mu_i,\widehat\mu_j)-\W(\mu_i,\mu_j)
\right|
\leq
\W(\widehat\mu_i,\mu_i)+\W(\widehat\mu_j,\mu_j).
\]
Therefore
\[
|\widehat U_n-U_n|
\leq
L\binom{n}{2}^{-1}\sum_{1\leq i<j\leq n}
\Bigl\{\W(\widehat\mu_i,\mu_i)+\W(\widehat\mu_j,\mu_j)\Bigr\}.
\]
Each term $\W(\widehat\mu_i,\mu_i)$ appears exactly $n-1$ times in the double sum, so
\[
\sum_{1\leq i<j\leq n}
\Bigl\{\W(\widehat\mu_i,\mu_i)+\W(\widehat\mu_j,\mu_j)\Bigr\}
=
(n-1)\sum_{i=1}^n \W(\widehat\mu_i,\mu_i).
\]
Since $\binom{n}{2}=n(n-1)/2$, it follows that
\[
|\widehat U_n-U_n|
\leq
\frac{2L}{n}\sum_{i=1}^n \W(\widehat\mu_i,\mu_i),
\]
which is the claimed bound.

\end{document}